\documentclass[b5paper,10pt,abstractoff,DIV=calc,headings=normal,twoside]{scrartcl}

\RequirePackage[OT1]{fontenc}
\RequirePackage{amsthm}
\usepackage[font=scriptsize]{caption}
\usepackage{soul}
\usepackage[normalem]{ulem}
\newcommand{\stkout}[1]{\ifmmode\text{\sout{\ensuremath{#1}}}\else\sout{#1}\fi}
\usepackage{comment}

\usepackage[textsize=tiny]{todonotes}

\usepackage{adjustbox}
\usepackage{amsfonts}
\usepackage{amssymb}
\usepackage[makeroom]{cancel}
\usepackage{bbold}
\usepackage{float}

\usepackage[colorlinks=true,
            linkcolor=blue,
            urlcolor=blue,
            citecolor=blue]{hyperref}

\usepackage{algorithm}
\usepackage{algpseudocode}
\usepackage{color}

\usepackage{etoolbox}

\AtBeginEnvironment{tabular}{\tiny}

\usepackage{csvsimple}

\newcommand{\B}[1]{\mathbf{#1}}
\newcommand{\M}{\mathfrak{m}}
\newcommand{\MS}{\cal M}

\AtBeginEnvironment{tabular}{\small}

\newcommand{\By}{\boldsymbol{\mathsf{y}}}

\newtheorem{theorem}{Theorem}

\definecolor{Florian}{RGB}{255,0,0}

\definecolor{Geir}{RGB}{0,0,255}

\definecolor{AH}{RGB}{255,0,255}

\usepackage[english]{babel}
\usepackage{lmodern}
\usepackage{microtype}
\usepackage{amsmath,amssymb,amsthm,amsfonts,graphicx,etoolbox,scrlayer-scrpage}
\usepackage[noblocks]{authblk}
\makeatletter
\patchcmd{\@maketitle}{\huge}{\Large}{}{}
\patchcmd{\abstract}{\quotation}{}{}{}
\AtBeginEnvironment{abstract}{\noindent\ignorespaces}
\AtEndEnvironment{abstract}{\par\mbox{}}

\newcommand{\shortauthor}{}
\newcommand{\shorttitle}{\@title}
\makeatother
\setkomafont{author}{\small}
\setkomafont{title}{\rmfamily\bfseries}
\setkomafont{disposition}{\rmfamily\bfseries}

\def\AMS#1{\par\noindent \textbf{AMS subject classification: }#1\par}

\newcommand{\keywords}[1]{\par\noindent\textbf{Keywords: }#1}
\subject{\vspace{-2\baselineskip}}\date{\vspace{-2\baselineskip}}
\rehead[]{\shortauthor}\lohead[]{\shorttitle}\lehead[]{}\rohead[]{}
\ofoot[]{}\ifoot[]{}\recalctypearea
\theoremstyle{plain}

\theoremstyle{definition}

\theoremstyle{remark}

\renewenvironment{abstract}{\bigskip\noindent\begin{minipage}{\textwidth}\setlength{\parindent}{15pt}\paragraph{Abstract:}}{\end{minipage}}

\begin{document}


\renewcommand{\shortauthor}{A. Hubin, F. Frommlet and G. Storvik}

\title{Reversible Genetically Modified Mode Jumping MCMC}

\author[1]{Aliaksandr Hubin\thanks{Corresponding author: aliaksah@math.uio.no. All authors thank NN9862K and NN9244K projects from \url{sigma2.no} for HPC resources provided.}}
\affil[1]{University of Oslo}
\author[2]{Florian Frommlet}
\affil[2]{Medical University of Vienna} 
\author[3]{Geir Storvik}
\affil[3]{University of Oslo} 
\maketitle

\begin{abstract}

In this paper, we introduce a reversible version of a genetically modified mode jumping Markov chain Monte Carlo algorithm (GMJMCMC) for inference on posterior model probabilities in complex model spaces, where the number of explanatory variables
is prohibitively large for classical Markov Chain Monte Carlo methods.
Unlike the earlier proposed GMJMCMC algorithm, the introduced algorithm is a proper MCMC and its limiting distribution corresponds to the posterior marginal model probabilities in the explored model space under reasonable regularity conditions. 
  
\end{abstract}


\smallskip

 \AMS{62-02, 62-09, 62F07, 62F15, 62J12, 62J05, 62J99, 62M05, 05A16, 60J22, 92D20, 90C27, 90C59} 


\section{Introduction}\label{sec:intro}
A genetically modified Markov chain Monte Carlo algorithm (GMJMCMC) was introduced \cite{hubin2018thesis, hubin2018deep, hubin2018novel} for Bayesian model selection/averaging problems when the total number of covariates (including functions of covariates) is prohibitively large. Applications include GWAS studies with Bayesian generalized linear models \cite{hubin2018thesis} as well as Bayesian logic regressions (BLR) \cite{hubin2018novel} and Bayesian generalized nonlinear models (BGNLM) \cite{hubin2018deep}. If certain regularity conditions are met, the GMJMCMC algorithm will asymptotically explore all models in the defined model spaces. However,  GMJMCMC is not a proper MCMC algorithm in the sense that its limiting distribution does not correspond to the marginal posterior model probabilities and thus only renormalized estimates of these probabilities \cite{Clyde:Ghosh:Littman:2010, hubin2018mode} can be obtained. In this paper, we introduce a reversible genetically modified Markov chain Monte Carlo algorithm (RGMJMCMC), which  modifies GMJMCMC to become a proper MCMC method
providing marginal posterior probabilities directly as Monte Carlo estimates.

\section{The algorithm}

\subsection{Genetically Modified MJMCMC}\label{sec:gmjmcmc}

Consider the case of a fixed predefined set of correlated features (covariates). Then the general model space $\MS$ is of size $2^q$ and standard MCMC algorithms tend to get stuck in local maxima if there are correlations between covariates \cite{Clyde:Ghosh:Littman:2010, hubin2018mode}. The basic idea of MJMCMC \cite{hubin2018mode} is to make a large jump (changing many model components) followed by local optimization within the discrete model space to obtain a proposal. The MJMCMC algorithm requires that all $q$ covariates defining the model space are potentially considered at each iteration of the algorithm. But if $q$ is large, it becomes impossible to specify and store all $2^q$ models in $\MS$. The idea behind GMJMCMC is to apply the MJMCMC algorithm iteratively to smaller sets of model components of size $s\ll q$. Here, $s$ is specified to be larger or equal to the maximal possible size $Q$ of the assumed true model in the defined model space, which is a necessary condition to be able to use GMJMCMC  \cite{hubin2018novel}. This constraint also reduces the number of models in the model space $\MS$ to  $\sum_{k=1}^{Q}{q \choose k}$. As shown in Theorem 1 in \cite{hubin2018novel}, GMJMCMC is irreducible in the defined model space of models of size up to $s$ under some easy to satisfy regularity conditions. Yet, GMJMCMC is not a proper MCMC and one cannot use the frequencies of different models in the Markov chain to estimate their posteriors. Instead, we use the renormalized estimates \cite{Clyde:Ghosh:Littman:2010, hubin2018mode, hubin2018novel, hubin2018deep} from a subspace $\MS^* \subset \MS$:
\begin{equation}
\widehat{p}(\M|\By) =  
\frac{p(\M)p(\By|\M)}
       {\sum_{\M' \in \MS^*}{p(\M')p(\By|\M')}}\,\text{I}(\M\in\MS^*) \;,
 \label{approxpost}
\end{equation}
which asymptotically converge to ${p}(\M|\By)$ as the number of iterations grows.

In GMJMCMC, we let $\mathcal{F}_0$ be all $q$ input features and $\mathcal{S}_0\subseteq\mathcal{F}_0$ be some subset of them. Then, throughout our search we generate a sequence of so called \emph{populations} $\mathcal{S}_1,\mathcal{S}_2,...,\mathcal{S}_{T_{max}}$. Each  $\mathcal{S}_t$ is a set of $s$ features and forms a separate \textit{search space} for exploration through MJMCMC iterations.
Populations dynamically evolve allowing GMJMCMC to explore different parts of the total model space.
Algorithm~2 in \cite{hubin2018novel} summarizes this procedure. The generation of $\mathcal{S}_{t+1}$ given $\mathcal{S}_t$ works as follows: 
Members of the new population $\mathcal{S}_{t+1}$ are generated by applying certain transformations to components of $\mathcal{S}_t$. 
First some components with low frequency from search space $\mathcal{S}_t$  are removed using a $filtration$ operator. The removed components are then replaced, where each replacement  is generated randomly by a \textit{mutation} operator with probability $P_m$, by a \textit{crossover}  operator with probability $P_c$, by a \textit{modification} operator with probability $P_t$ or by a \textit{projection} operator with probability $P_p$, where $P_c+P_m+P_t+P_p = 1$. The operators to generate potential features of $\mathcal{S}_{t+1}$ are formally defined in \cite{hubin2018thesis, hubin2018novel, hubin2018deep}.

\subsection{Reversible Genetically Modified MJMCMC}

The GMJMCMC algorithm described above cannot guarantee that the ergodic distribution of its Markov chain corresponds to the target distribution of interest \cite{hubin2018novel}. An easy modification based on performing both forward and backward swaps between populations can provide a proper MCMC algorithm in the model space of interest.
Consider a  transition $\M \rightarrow \mathcal{S}' \rightarrow \M'_0 \rightarrow ... \rightarrow \M'_k \rightarrow \M'$ with a given probability kernel. Here, $q_S(\mathcal{S}'|\M)$  is the proposal for a new population, 
transitions  $\M'_0 \rightarrow ... \rightarrow \M'_k$ are generated by local MJMCMC within the model space induced by $\mathcal{S}'$, and the transition $\M'_k \rightarrow \M'$ is some randomization at the end of the procedure as described in the next paragraph. 
The following theorem shows the detailed balance equation for the suggested swaps between models.
\begin{theorem}
Assume  $\M\sim p(\cdot|\B y)$ and $(\mathcal{S'},\M_k',\M')$ are generated according to the proposal distribution $q_S(\mathcal{S'}|\M)q_o(\M_k'|\mathcal{S'},\M)q_r(\M'|\mathcal{S},\M_k')$. Assume
further $(\mathcal{S},\M_k)$ are generated according to $\tilde{q}_S(\mathcal{S}|\M',\mathcal{S},\M)q_o(\M_k|\mathcal{S},\M')$. Let
\[
\M^*=\begin{cases}
\M'&\text{with probability $\min\{1,a_{mh}\}$;}\\
\M&\text{otherwise.}
\end{cases}
\]
where
\begin{equation}\label{eq:a.rjmcmc}
a_{mh} = \frac{p(\M'|\B y)q_r(\M|\mathcal{S},\M_k)}{p(\M|\B y)q_r(\M'|\mathcal{S'},\M'_k)}.
\end{equation}
Then  $\M^*\sim p(\cdot|\B y)$.
\end{theorem}

\begin{proof}
Define $\bar{p}(\M,\mathcal{S'},\M_k')\equiv p(\M|\B y)q_\mathcal{S}(\mathcal{S'}|\M)q_o(\M_k'|\mathcal{S'},\M)$. Then by construction $(\M,\mathcal{S'},\M_k')\sim \bar{p}(\M,\mathcal{S},\M_k')$.
Define $(\M',\mathcal{S},\M_k)$ to be a proposal from the distribution  $q_r(\M'|\mathcal{S},\M_k')q_S(\mathcal{S}|\M')q_o(\M_k|\mathcal{S},\M')$.
Then the Metropolis-Hastings acceptance ratio becomes
\[
\frac{\bar{p}(\M',\mathcal{S},\M_k)q_r(\M|\mathcal{S},\M_k)q_S(\mathcal{S'}|\M)q_o(\M_k'|\mathcal{S'},\M)}{\bar{p}(\M,\mathcal{S'},\M_k')q_r(\M'|\mathcal{S}',\M_k')q_S(\mathcal{S}|\M')q_o(\M_k|\mathcal{S},\M')}
\]
which reduces to $a_{mh}$. 
\end{proof}

From Theorem 1, it follows that if the Markov chain is irreducible in the model space then it is ergodic and converges to the right posterior distribution. The described procedure marginally generates samples from the target distribution, i.e. the posterior model probabilities $p(\M|\By)$. Instead of using the approximation~\eqref{approxpost}  one can get frequency-based estimates of the model posteriors $p(\M|\By)$. For a sequence of simulated models $\M^1,...,\M^W$ from an ergodic MCMC algorithm with a stationary distribution $p(\M|\By)$ it holds that
\begin{equation}\label{map2}
\widetilde{p}(\M|\By)=W^{-1}\sum_{i=1}^W \text{I}(\M^{(i)} = \M) \xrightarrow[W\rightarrow\infty]{d} p(\M|\By), \; 
\end{equation}
and similar results are valid for estimates of the posterior marginal inclusion probabilities or any other parameters of interest \cite{Clyde:Ghosh:Littman:2010, hubin2018mode}.

In practise, proposals $q_\mathcal{S}(\mathcal{S}'|\M)$  are obtained as follows: First all members of $\M$ are included. Then additional
features are added by the same operators as described in Section~\ref{sec:gmjmcmc} but with $\mathcal{S}_t$ replaced by the population including all components in $\M$. The randomization $\M'\sim q_r(\M|\mathcal{S}',\M_k')$ is performed by potential swapping of the features within $\mathcal{S}'$, each with a small probability $\rho_r$.
Note that this might give a reverse probability
$q_r(\M|\mathcal{S},\M_k)$ being zero if $\mathcal{S}$ does not include all components in $\M$. In that case the proposed model is not accepted. Otherwise 
the ratio of the proposal probabilities can be written as $\frac{q_r(\M|\mathcal{S},\M_k)}{q_r(\M'|\mathcal{S}',\M'_k)} = \rho_r^{d(\M,\M_k)-d(\M',\M'_k)}$, where $d(\cdot,\cdot)$ is the Hamming distance (the number of components differing).
\subsubsection*{Delayed rejection}
The computationally most demanding parts of the RGMJMCMC algorithm are the forward and backward MCMC searches. Often, the proposals generated by forward search have a very small probability $\pi(\M')$ resulting in a low acceptance probability regardless of the way the backward auxiliary variables are generated. In such cases, one would like to reject directly without performing the backward search. This is achieved by the delayed acceptance procedure~\cite{banterle2015accelerating} which can be applied in our case due to the following result:
\begin{theorem}
Assume $\M\sim p(\cdot|\B y)$ and $\M'$ is generated according to the RGMJMCMC algorithm. Accept $\M'$ if
both
\begin{enumerate}
\item $\M'$ is preliminarily accepted with a probability $\min\{1,\tfrac{p(\M'|\B y)}{p(\M|\B y)}\}$,
\item $\M'$  is finally accepted with a probability 
$\min\{1,\tfrac{q_r(\M|\mathcal{S},\M'_k)}{q_r(\M'|\mathcal{S}',\M_k)}\}$.
\end{enumerate}
Then also $\M \sim p(\cdot|\B y)$.
\end{theorem}

\begin{proof}
It holds for $a_{mh}$ given by~\eqref{eq:a.rjmcmc} that
\begin{align*}
a_{mh}(\M,\mathcal{S}',\M_k';\M',\mathcal{S},\M_k)
=&  a_{mh}^1(\M,\mathcal{S}',\M_k';\M',\mathcal{S},\M_k)\times a_{mh}^2(\M,\mathcal{S}',\M_k';\M',\mathcal{S},\M_k)
\intertext{where}
a_{mh}^1(\M,\mathcal{S}',\M_k';\M',\mathcal{S},\M_k)=&\frac{p(\M'|\B y)}{p(\M|\B y)},\quad
a_{mh}^2(\M,\mathcal{S}',\M_k';\M',\mathcal{S},\M_k)=\frac{q_r(\M|\mathcal{S},\M'_{k})}{q_r(\M'|\mathcal{S}',\M_k)}
\end{align*}
Since $a_{mh}^j(\M,\mathcal{S}',\M_k';\M',\mathcal{S},\M_k)
=[a_{mh}^j(\M',\mathcal{S},\M_k;\M,\mathcal{S},\M'_k)]^{-1}$ for $j=1,2$, by the general results of~\cite{banterle2015accelerating} we obtain an invariant kernel for the target.
\end{proof}

\section{Applications}\label{sec:applications}

We repeat the experiments from \cite{hubin2018deep} concerned with recovering the planetary mass law (I), the 3rd Kepler's law (II), and a logic regression example (III). We follow the original experimental design and refer the reader to \cite{hubin2018deep} for full detail. The parallelization strategy is described in \cite{hubin2018novel, hubin2018deep}. The number of threads is denoted by T in the following tables.  Table~\ref{tab:laws} reports results for datasets (I) and (II). Power, FDR, and the expected number of FP are estimated based on 100 runs of RGMJMCMC and GMJMCMC for each number of threads, respectively. The set of nonlinearities $\mathcal{G}_2$ from \cite{hubin2018deep} is used in these experiments. Similarly, Table~\ref{tab:logic} gives results for the logic regression example (III), where the set of nonlinearities $\mathcal{G}_1$ from \cite{hubin2018deep} is used for BGNLM.
\begin{table}[H]
\centering
\caption{Overall Power, average number of false positives (FP) and FDR for detecting the planetary mass law (planetary mass as response) and 3rd Kepler's law (semi-major axis as response). For the former, only $R^3_p\times \rho_p$ is considered as a TP discovery. For the latter,  $F_1 = \left(P\times P \times  M_h\right)^{\frac{1}{3}}$,  $F_2 = \left(P\times P \times  R_h \right)^{\frac{1}{3}}$  and  $F_3 = \left(P \times P \times T_h \right)^{\frac{1}{3}}$ are counted as TPs. Results for GMJMCMC are given in parentheses.} 

\label{tab:laws}
\begin{tabular}{lccc|lccc}
\multicolumn{4}{c}{Planetary mass law} &\multicolumn{4}{c}{3rd Kepler's law}\\\hline 
T&Power&FP&FDR&T&Power& FP&FDR\\\hline
16&0.94 (0.93)&0.29 (0.36)&0.18 (0.22)&64&1.00 (0.99)&0.04 (0.04)&0.02 (0.02) \\
4&0.63 (0.69)&0.64 (0.49)&0.38 (0.34)&16&0.65 (0.83)&0.88 (0.55)&0.39 (0.22)\\
1&0.29 (0.42)&1.54 (1.25)&0.71 (0.58)&1&0.06 (0.14)&2.14 (1.81)&0.94 (0.86)\\\hline
\end{tabular}
\end{table}

Tables~\ref{tab:laws} and \ref{tab:logic} illustrate that in all three applications the RGMJMCMC algorithm is capable of estimating the posterior marginal probabilities of different features, and thus to recover the true data generative processes, with reasonable Power and FDR. The performance is on par with GMJMCMC for datasets (I) and (II). For dataset (III), it is a bit worse than GMJMCMC for BGNLM but slightly better than BLR (see logic regression example in \cite{hubin2018deep}). 
\begin{table}[H]
\centering
\caption{Power for individual trees, overall power,  expected number of FP, and FDR are compared between RGMJMCMC (R), GMJMCMC (G) and Bayesian Logic regression (L). Logic expressions from the data generating model are $L_1=X_7$, $L_2=X_8$, $L_3=X_2*X_9$, $L_4=X_{18}*X_{21}$, $L_5=X_{1}* X_{3}* X_{27}$, $L_6=X_{12}* X_{20}* X_{37}$, $L_7=X_{4}* X_{10}* X_{17}* X_{30}$, $L_8=X_{11}* X_{13}* X_{19}* X_{50}$. } 
\label{tab:logic}
\small
\begin{tabular}{lcccccccccccc}
\hline
&T&$L_1$& $L_2$& $L_3$& $L_4$& $L_5$& $L_6$& $L_7$& $L_8$& Power& FP& FDR\\
\hline
R&32&1.00&1.00&0.96&1.00&1.00&1.00&0.92&0.89&0.97&1.14&0.13\\
G&32&1.00&1.00&1.00&1.00&1.00&1.00&0.99&0.98&1.00&0.51&0.06\\
L&32&0.99&1.00&1.00&0.96&1.00&0.99&0.91&0.38&0.90&1.09&0.13\\
\hline

\end{tabular}
\end{table}
\section{Discussion}
In this paper, we have introduced RGMJMCMC and proved its theoretical properties. We have also repeated some experiments described in \cite{hubin2018deep} to assess RGMJMCMC in terms of model identification. In an extended publication, it would be of interest to additionally evaluate the accuracy of posterior estimates and the performance in terms of out-of-sample prediction accuracy.





\bibliographystyle{abbrv} 
\bibliography{ref}

\end{document}